\documentclass[12pt]{article}

\usepackage{hyperref}
\usepackage[english]{babel}
\usepackage{amssymb,amsmath, amsthm, amscd,ifthen}
\usepackage{graphicx}
\usepackage{xcolor}
\usepackage{dsfont}
\usepackage{algorithm}
\usepackage{comment}
\usepackage[margin=1.0in]{geometry}
\usepackage[noend]{algpseudocode}
\makeatletter
\def\BState{\State\hskip-\ALG@thistlm}
\makeatother

\newcommand{\R}{{\mathbb R}}
\newcommand{\N}{{\mathbb N}}
\newcommand{\RR}{{\mathcal R}}
\newcommand{\FF}{{\mathcal F}}
\newcommand{\eps}{\epsilon}
 \newcommand{\E}{\mathbb{E}}
 \newcommand{\ind}{\mathds{1}}

\newcommand{\PP}{{\mathrm P}}

\newcommand{\X}{{\mathcal X}}
\newcommand{\ff}{{\mathcal F}}

\newtheorem{thm}{Theorem}
\newtheorem{obs}[thm]{Observation}
\newtheorem{lem}[thm]{Lemma}
\newtheorem{cor}[thm]{Corollary}

\date{}

\newtheorem{thmx}{Theorem}
\newtheorem{lemm}[thmx]{Lemma}

\title{When are epsilon-nets small?}
\author{Andrey Kupavskii\thanks{Moscow Institute of Physics and Technology and University of Oxford. The research of Andrey Kupavskii was supported by the grant of Russian Government N 075-15-2019-1926. \href{mailto:kupavskii@yandex.ru}{kupavskii@yandex.ru}} \and Nikita Zhivotovskiy\thanks{Higher School of Economics. Now at Google Research, Brain Team. Part of this research was done when Nikita Zhivotovskiy was a postdoctoral fellow at Technion. The reported study was funded by
RFBR according to the research project N 18-37-00-371. \href{mailto:nikita.zhivotovskiy@phystech.edu}{nikita.zhivotovskiy@phystech.edu}}}
\date{}

\begin{document}

\maketitle

\begin{abstract} Given a {\it range space} $(\X,\RR)$, where $\X$ is a set equipped with probability measure $\PP$, $\RR\subset 2^{\X}$ is a family of measurable subsets,  and $\eps>0$, an {\it $\eps$-net} is a subset of $\X$ in the support of $\PP$, which intersects each $R\in\RR$ with $\PP(R)\ge \eps$.  In many interesting situations the size of $\eps$-nets depends on $\eps$  and on natural complexity measures. The aim of this paper is to give a systematic treatment of such complexity measures arising in  Discrete and Computational Geometry and Statistical Learning, and to bridge the gap between the results appearing in these two fields. As a byproduct, we obtain several new upper bounds on the sizes of $\eps$-nets that generalize/improve the best known general guarantees. Some of our results deal with improvements in logarithmic factors (which is a subject of several classical problems in Learning Theory and Computational Geometry), while others consider the regimes where  $\eps$-nets of size $o(\frac{1}{\eps})$ exist. Inspired by results in Statistical Learning, we also give a short proof of the Haussler's upper bound on packing numbers \cite{Haussler95}.
\end{abstract}

%\begin{keyword}
 % Epsilon-nets, active learning, covering numbers, Alexander's capacity, Shallow cell complexity
%\end{keyword}
%\end{frontmatter}

\section{Introduction}
In this section, we collect $\eps$-nets-related results in different areas. Some of the observations that we make are scattered in the literature and sometimes not written explicitly, which is one reason why we collected them in a single section. The new results are presented starting from Section~\ref{sec2}.

\subsection{$\eps$-nets. Combinatorial and geometric point of view}
Consider a set $\X$ equipped with probability measure $\PP$ and a family of measurable subsets $\RR\subset 2^{\X}$, where $2^{\X}$ is the power set of $\X$. For simplicity and to avoid potential measurability issues we assume that the set $\X$ is finite. The pair $(\X,\RR)$ is called a {\it range space}. For a fixed $\eps>0$, a subset $S\subset \X\cap\textrm{supp}(\PP)$ is an {\it $\eps$-net},\footnote{Compare with the definition of a \emph{weak $\eps$-net} where $S\subset \X\cap\textrm{supp}(\PP)$ is replaced by $S \subset \X$. In what follows $\textrm{supp}$ denotes the measure theoretic support of $\PP$.}  if $R\cap S\ne \emptyset$ for each $R\in\RR$ with  $\PP(R)\ge \eps$. $\eps$-nets often arise in the context of Computational Geometry problems. In that context, the set $\X$ is often finite and equipped with the uniform measure, and we may speak about {\it sizes} instead of {\it measures}. In some cases it is not difficult to generalize the results that are valid for a uniform measure on a finite set to an arbitrary measure (see e.g., \cite{Matousek}). However, to avoid some technical difficulties, in this paper we focus on distribution dependent complexity measures of the range spaces and therefore choose to present everything for general probability measures.

A line of research, starting from a seminal work of Vapnik and Chervonenkis (we refer to the textbook \cite{VC}), is concerned with proving the existence of small $\eps$-nets in certain scenarios. The first historically and probably the most important one is {\it bounded VC-dimension}. The {\it VC-dimension} of $(\X, \mathcal{R})$ is the maximal size of  $Y \subseteq \X$ such that $Y$ is {\it shattered}: $\mathcal{R}|_{Y}:=\{R\cap Y: R\in \RR\} = 2^{Y}$. This means that any subsets of $Y$ is realized by some $R\in \RR$. Haussler and Welzl \cite{HW} proved that one can find $\eps$-nets of size depending on $\eps$ and VC-dimension only. Moreover, the $\eps$-net can be constructed using an i.i.d. sample of points taking their values in $\X$. Later their result was slightly sharpened in \cite{KPW}, and matching lower bounds were proven. We note that the result of Haussler and Welzl follows immediately from  the uniform ratio-type bounds in \cite{VC}. Given the VC-dimension of a range space, one can bound the size of the smallest $\eps$-net as follows.

\begin{thmx}[\cite{VC}, \cite{HW}, \cite{KPW}]\label{thmvc}
Fix $\eps, \delta \in (0, 1]$ and let $(\X, \mathcal{R})$ be a range space of VC-dimension $d$ with probability distribution $\PP$. Then a set of size $m = O\big( \frac{d\log\frac{1}{\eps}}{\eps} + \frac{\log\frac{1}{\delta}}{\eps}\big)$ chosen i.i.d. from $\X$ according to $\PP$ is an $\eps$-net for $(\X,\RR)$ with probability at least $1 - \delta$.
\end{thmx}

{ It can be more convenient to formulate this and analogous results for given sample size: what is the best $\eps$ that can be obtained with a given budget of $n$ points}. Here is such a reformulation of  Theorem \ref{thmvc}, equivalent to the original.

\begin{thmx}\label{thmvcrev}
Fix $\delta \in (0, 1]$, $n\in \mathbb N$ and let $(\X, \mathcal{R})$ be a range space of VC-dimension $d$ with probability distribution $\PP$.
A set of size $n$ chosen i.i.d. from $\X$ according to $\PP$ is an $\eps(n)$-net for $(\X,\RR)$ with probability at least $1 - \delta$ for $\eps(n) = O\big( \frac{d\log\frac{n}{d}}{n} + \frac{\log\frac{1}{\delta}}{n}\big)$.
\end{thmx}

In Computational Geometry $\X$ is typically a set of points in $\R^d$, and the ranges are intersections of $\X$ with all objects from a certain class: lines, halfspaces, balls, etc. One then searches for upper bounds on the sizes of  $\eps$-nets that would hold for {\it all} range spaces of such type. In another common scenario of the so-called {\it dual range spaces}, the roles of points and ranges are switched.

The applications of $\eps$-nets cover several topics in Computational Geometry, including spatial partitioning and LP rounding. We refer to a recent survey \cite{MV}, which covers many of the recent developments in $\eps$-nets, as well as their applications.

\subsection{$\eps$-nets. Statistical point of view}
Similar ideas and notions were developed in Statistical Learning. Consider the following statistical model. We are given an instance space $\mathcal{X}$ equipped with an unknown probability distribution $\PP$ and a (known) family of classifiers ({ binary valued functions}) $\mathcal{F}$ consisting of functions $f: \mathcal X \to \{\pm1\}$.
 A learner observes $\left((x_{1}, y_{1}), \ldots, (x_{m}, y_{m})\right)$, an i.i.d. training sample where $x_i$ are sampled according to $\PP$ and $y_i = f^*(x_i)$ for some fixed $f^* \in \mathcal{F}$. This scenario is referred to as the {\it realizable case} classification and the learning model itself is referred to as \emph{Passive learning}.
 \emph{Sample-consistent learning
algorithm} (the particular case of empirical risk minimization) refers to any learning algorithm with the following property: given a training sample of size $m$, it outputs any classifier $\hat{f}\in \mathcal{F}$ that is consistent with the sample (that is, $\hat f(x_i)=f^*(x_i)$ for all $i=1,\ldots, m$).

Similarly as above, we say a set $\{x_1, \ldots, x_k\} \in \mathcal{X}^{k}$ is {\it shattered} by $\mathcal{F}$ if there are all $2^k$ distinct classifications of $\{x_1, \ldots, x_k\}$ realized by classifiers in $\mathcal{F}$. The {\it VC-dimension} of $\mathcal{F}$ is the largest integer $d$ such that there exists a set $\{x_1, \ldots, x_d\}$ shattered by $\mathcal{F}$.

 The analogue of Theorem~\ref{thmvc} in this context is the following classical result\footnote{Although this result is not presented explicitly in their book, it follows directly from their learning bounds for the realizable case classification.} of Vapnik and Chervonenkis:

\begin{thmx}[Theorem 12.2 in \cite{VC}]
\label{vcbound}
Consider any sample-consistent learning algorithm over the i.i.d. sample of size $m$, which outputs a classifier $\hat f\in\ff$.
Assume that $\mathcal{F}$ has VC dimension $d$ and we are in the realizable case. Then, for any $\eps, \delta \in (0, 1]$ and some $m(\eps) = O\big( \frac{d\log\frac{1}{\eps}}{\eps} + \frac{\log\frac{1}{\delta}}{\eps}\big)$, we have $\PP\big(\hat{f}(x) \neq f^*(x)\big) \le \eps$ with probability at least $1 - \delta$.
\end{thmx}

It is actually easy to translate this problem into the combinatorial language: we just have to think of the instance space $\mathcal X$ as our ground set, and the collection of sets $\big\{\{x \in \X: f(x)\ne f^*(x)\}: f\in\ff\big\}$ playing the role of ranges. Then Theorem~\ref{vcbound}, basically says that an i.i.d. sample gives an $\eps$-net for such range space with high probability.\\

\subsection{Alexander's capacity and Active Learning}
Significant effort was put by many researchers in both Computational Geometry and Statistical Learning Theory to understand whether it is possible to improve on the above bounds. In this context, several different measures of complexity were introduced. One of them is the VC-dimension, already mentioned above. Actually, to Prove Theorem~\ref{thmvc}, one only needs the following property of the VC-dimension, implied by the famous Vapnik-Chervonenkis-Sauer-Shelah lemma: given a range space $(\X,\RR)$ of VC-dimension $d$, for any $Y\subset \X$ we have $|\RR|_Y|\le \sum_{i=0}^d{|Y|\choose d} \le \left(\frac{e|Y|}{d}\right)^d$. Let us introduce the {\it projection function} $\pi_{\RR}(Y)$:
\begin{equation}
\label{shatter}
\pi_{\RR}(y):=\max\{|\RR|_Y|:Y\subset \X, |Y|=y\}.
\end{equation}
Thus, the Vapnik-Chervonenkis-Sauer-Shelah lemma \cite{VC} implies that
\begin{equation}
\label{vcs}
\pi_{\RR}(y) \le \sum\limits_{i = 0}^{d}{y \choose i} \le \left(\frac{ey}{d}\right)^d,
\end{equation} for any $y=d,\ldots,|\X|$, and Theorem~\ref{thmvc}, as well as Theorem~\ref{vcbound}, holds under this condition. Actually, Vapnik and Chervonenkis used a weaker requirement to obtain Theorem~\ref{vcbound}: They required the projections to be small {\it on average} (see \cite{Boucheron05, Bousquet05} and the results related to the so-called \emph{VC entropy}).

One of the measures  coming from  Statistical Learning is Alexander's capacity \cite{Alexander87, Gine06}. Initially it appeared in the work of Alexander \cite{Alexander87} in the analysis of ratio-type empirical processes. For $\eps_{0} > 0$ fix a set $\mathcal{F}_{\eps_{0}} := \bigl\{f \in \mathcal{F} : \PP\bigl(f(x) \neq f^{*}(x)\bigr)\leq \eps_0\bigr\}$.
For $\eps \in (0,1]$, define {\it Alexander's capacity}\footnote{Instead of defining it with respect to the worst $f^* \in \ff$ we work with the definition that depends on the target function $f^*$. This does not affect the results.} $\tau(\eps)$ as follows.
\[
\tau(\eps) :=
\sup\limits_{\eps_{0} \ge \eps}
\frac{\PP(\{x \in \mathcal{X} : \exists f\in \mathcal{F}_{\eps_{0}} \text{ s.t. } f(x) \neq f^{*}(x)\})}
{\eps_{0}}.
\]

{ \textbf{Remark.}  Alexander's capacity is essentially the same as the {\it disagreement coefficient} in the Active Learning literature \cite{Hanneke14}. The difference is that in the original work of Alexander \cite{Alexander87} the ratio $\frac{\PP(\{x \in \mathcal{X} : \exists f\in \mathcal{F}_{\eps_{0}} \text{ s.t. } f(x) \neq f^{*}(x)\})}
{\eps_{0}}$ is assumed to be non-decreasing with $\eps_{0}$. We avoid this technical assumption by taking the supremum with respect to $\eps_{0}$ as in \cite{Hanneke14}.}

In what follows, to avoid problems with logarithms, we separate $\tau(\eps)$ from zero by an absolute constant. Without the loss of generality, we redefine the capacity by $\max\{\tau(\eps), 1\}$.

We  also define {\it Alexander's capacity} for a range space $(\X,\RR)$ as follows
\[
\tau(\eps) := \sup_{\eps_0\ge \eps} \frac{\PP\big(\bigcup_{R \in \mathcal{R}_{\le \eps_{0}}} R\big)}{\eps_0},
\]
where $ \mathcal{R}_{\le \eps_{0}} := \{R \in \mathcal{R} : \PP(R) \le \eps_{0}\}$.
For the uniform measure on a finite set the last definition can be informally understood as ratio of the number of points of $\X$ that lie in one of the sets of size at most $\eps_0 n$, over $\eps_0n$ (which is maximized over $\eps_0\ge \eps$).
{ Before we proceed let us point some other trivial properties of the Alexander's capacity, which will be used below extensively. Observe that $\tau(\eps)\le \frac{1}{\eps}$ and $\epsilon \le \epsilon^{\prime}$ implies that $1 \le \frac{\tau(\epsilon)}{\tau(\epsilon^{\prime})} \le \frac{\epsilon^{\prime}}{\epsilon}$. Denoting $\tau_i = \tau(2^i\epsilon)$, we also have $\tau_i \le \frac{2^{-i}}{\epsilon}$, and $1 \le \frac{\tau_i}{\tau_{i+1}} \le 2$, as well as $\sum\limits_{i = 1}^{\infty}\tau_i \le 1/\epsilon$.
}

\begin{thmx}[Gine and Koltchinskii \cite{Gine06}, Hanneke \cite{Hanneke14}]
\label{Koltch}
Theorems \ref{thmvc} and \ref{vcbound} hold with a sample of size
$
m = O\big(\frac{d\log\tau(\eps) + \log\frac{1}{\delta}}{\eps}\big).
$
\end{thmx}

Observe that, since $\tau(\eps) \le \frac{1}{\eps}$, Theorem \ref{Koltch} is an improvement over Theorem \ref{vcbound}. We refer to \cite{Balcan13, Gine06, Hanneke14} where many examples with $\tau(\eps) = o(\frac{1}{\eps})$ are provided. The same result may be directly translated to the range spaces. However, we show that in many cases where $\tau(\eps)$ is smaller than $\frac{1}{\eps}$ it is possible to construct nets of sizes significantly smaller than what is guaranteed by Theorem \ref{Koltch}. The result of Theorem \ref{Koltch} is very specific to i.i.d. sampling and does not cover the situation where one is able to choose points in a more clever way.

\emph{Active learning} is a particular framework within Statistical Learning. As before, there is an instance space $\mathcal{X}$ and a label space $\mathcal{Y} := \{-1, 1\}$ and a set $\mathcal{F}$ of classifiers mapping $\mathcal{X}$ to $\mathcal{Y}$. In the the realizable case, there is a target function $f^* \in \mathcal{F}$ and a sample $(x_i, f^*(x_i))_{i = 1}^n$. In the pool-based active learning, we define an active learning algorithm as
an algorithm taking as input a budget $n \in \mathbb{N}$, and proceeding as follows (compare with \emph{passive learning} described in the previous section). The algorithm
initially uses an unlabeled independent { infinite} data sequence $x_1, x_2, \ldots$ distributed according to $\PP$. At time $i$, the algorithm may select an index $i_1$ and request to observe the label $y_{i_1}:=f^*(x_{i_1})$. The algorithm  observes the value of $y_{i_1}$, and if $n \ge 2$, then based on both the unlabeled sequence and $y_{i_1}$, it may select another index $i_2$ and request to observe $y_{i_2}$. This continues for at most $n$ rounds.

In the realizable case the algorithm named CAL (named after Cohn, Atlas, and Ladner \cite{Chon92}) is by now the most studied. We refer the reader to \cite{Hanneke14} for the exact definition and the analysis of this learning procedure (cf. also Appendix~\ref{cal}).

\begin{thmx}[Sample complexity bound for CAL \cite{Hanneke14}]
\label{alexcapcal}
The exists an active learning algorithm (CAL), such that in the realizable case, for any distribution $\PP$, $\eps, \delta \in (0, 1]$ and a class of classifiers $\mathcal{F}$ of VC dimension $d$, after requesting
\[
n = O\Big(\tau(\eps)\big[d\log\tau(\eps) + \log\log\frac 1{\eps} + \log\frac1{\delta}\big]\log\frac 1{\eps}\Big),
\]
labels with probability at least $1 - \delta$, CAL returns a classifier $\hat{f}$ satisfying $\PP(\hat{f}(x) \neq f^*(x))\le\eps$.
\end{thmx}
If $\tau(\eps)\log \tau(\eps) = o(1/\eps)$ the above algorithm allows for up to an \emph{exponential improvement} over the standard sampling strategy described in Theorem \ref{vcbound}. In particular, a simple inspection of the proof of the result for the CAL algorithm implies the following bound for $\eps$-nets. For the sake of completeness an analog of CAL algorithm for range spaces together with some additional discussions are presented in Appendix \ref{cal}.
\begin{cor}
\label{alexcap}
Let $(\X, \mathcal{R})$ be a range space of VC-dimension $d$ and Alexander's capacity $\tau(\eps)$.
Then there exists an $\eps$-net of size
\[
n = O\Big(\tau(\eps)\big[d\log\tau(\eps) + \log\log\frac 1{\eps}\big]\log\frac 1{\eps}\Big).
\]
\end{cor}

\subsection{Teaching}
Teaching is another important topic in Learning Theory. There are several natural frameworks of teaching studied in the literature (we refer to \cite{Doliwa14, Goldman95} and reference therein). Once again, for simplicity we consider the case of a finite domain $\X$. One basic teaching framework \cite{Goldman95} may be described as follows: a helpful teacher who is aware of the target function $f^*$ selects a sample $S = \left((x_{1}, f^*(x_1)), \ldots, (x_{m}, f^*(x_m))\right)$ and presents it to the learner who is aware only of $\mathcal{F}$. The aim of the teacher is to present a small sample $S$ with the following property: $f^*$ is the only function in $\mathcal{F}$ that is consistent with $S$. { We return to Teaching in Section \ref{sec6} in the contexts of its relations with $\eps$-nets.}

\subsection{Structure of the paper and our contributions}
\begin{enumerate}
\item In Section~\ref{sec2} we provide a simple $\eps$-net version of Theorem \ref{alexcapcal} with better guarantees.
\item In Section~\ref{sec3} we recall the definition of the {\it doubling constant} and give some of the results concerning it along with their improvements. Later, we prove a general theorem (which is one of our main contributions) giving bounds on $\eps$-nets in terms of both Alexander's capacity and the doubling constant, which improves upon many of the previously known results, and show that it is tight in all regimes. In Section~\ref{sec5} we return to the doubling constant and provide sharp bounds on it.
\item In Section \ref{sec3.5} below we give a transparent proof of Haussler's packing Lemma \cite{Haussler95}. Matou\u{s}ek \cite{Matousek} remarked that the original proof of Haussler (as well as its version by Chazelle \cite{Chaz92})  uses a ``probabilistic
argument which looks like a magician's trick". It is surprising that our proof uses only the techniques available in the literature at the time of original publication of Haussler. It is useful to note that recent Computational Geometry papers (see e.g., \cite{Dutta15, Mustafa16}) still use the argument of the original proof of Haussler. Apart from a simplified proof, our result may also be used to provide sharper constant factors and will be used in Section \ref{sec5}.
\item Section \ref{sec6} is devoted to several concluding remarks.
\end{enumerate}

\section{New bound for $\eps$-nets in terms of Alexander's capacity}
\label{sec2}
The following theorem can be quickly deduced via an application of Theorem~\ref{thmvc}.

\begin{thm}\label{thmmain2} Let $(\X, \RR)$ be a range space of VC-dimension $d$. Fix $\epsilon>0$. Let $\tau_i:=\tau(2^i\eps)$ and put $z:= { 1 +} \lceil \log_2 \frac 1{\eps}\rceil$.
Then there exists an $\eps$-net for $(\X,\RR)$ of size \begin{equation}\label{eqmain2}O\Big(d\sum_{i=1}^{z}\tau_i\log \tau_i\Big).\end{equation}
\end{thm}

\begin{proof}
We set $\eps_{-1}:=0$, $\eps_{i} := 2^i\eps$ and $\mathcal{R}_{i} = \{S \in \mathcal{R} : \eps_{i-1} \le \PP(S) <\eps_{i}\}$ for $i = 0, 1,\ldots $. Define $\X^{(i)}:=\bigcup_{R\in \RR_i}R$ to be the support of $\RR_i$.

It is sufficient to find an $\epsilon$-net for each $(\X^{(i)},\RR_i)$ of size $O(d \tau_i\log \tau_i)$. Note that, by definition, $\PP(\X^{(i)})\le \tau_i \eps_i $, while for each $R\in \RR_i$ we have $\PP(R)\ge \eps_{i-1}$. Therefore, $\PP(R)\ge \frac {\PP(\X^{(i)})}{2\tau_i}$, and since for each $R\in \RR_i$ we have $\PP(R \cap \X^{(i)}) =\PP(R)$ and thus $\PP(R|\X^{(i)}) =\PP(R)/\PP(\X^{(i)}) \ge \frac{1}{2\tau_i}$ it is sufficient for us to find a $\frac 1{2\tau_i}$-net for $(\X^{(i)},\RR_i)$ with respect to the conditional distribution $\PP(\ |\X^{(i)})$. But this could simply be done using the Vapnik-Chervonenkis-Haussler-Welzl Theorem \ref{thmvc}. This gives a net of size $O(d \tau_i\log \tau_i)$ for each $(\X^{(i)},\RR_i)$.
\end{proof}

We immediately get the following Corollary.
\begin{cor}\label{coralexnew}
   In the notation of Theorem~\ref{thmmain2}, there exists an $\eps$-net for $(\X,\RR)$ of size
   \begin{equation}
   \label{eqtau}
   O\Big(d \tau(\eps)\log \tau(\eps)\log \frac 1{\eps} \Big).
   \end{equation}
\end{cor}
\begin{proof} There are $1 + \lceil \log_2 \frac 1{\eps}\rceil$ summands in  \eqref{eqmain2}, each being $O\big({d}\tau(\eps)\log \tau(\eps)\big)$.
\end{proof}

\textbf{Remark. } Theorem~\ref{thmmain2} improves on both Theorems~\ref{thmvc} and~\ref{Koltch} in many cases. Indeed, one should use $\tau_i\le 2^{-i}/(\eps)$.  Corollary~\ref{coralexnew} is an improvement of Corollary~\ref{alexcap}.
In particular, Corollary~\ref{coralexnew} implies that if $\tau(\eps) = O(1)$ we have $\eps$-nets of size $O(\log \frac{1}{\eps})$ which is significantly smaller than what is guaranteed by Theorem \ref{thmvc}. { Some particular examples will be given in what follows.}

\vskip+0.1cm

The $\eps$-net from Theorem~\ref{thmmain2} is based on a {\it simple sampling} strategy, although the probability of including different elements differs. The probabilities can be decided on before choosing a random sample quite easily. One should just find the sets $\mathcal{X}^{(i)}$, which could be done efficiently in some cases. However, Theorem~\ref{thmvc}, as well as CAL algorithm used for Corollary~\ref{alexcap}, gives a more natural sampling strategy to construct $\eps$-nets.\vskip+0.1cm

\section{The doubling constant}\label{sec3}

Another quantity of interest with a wide range of applications in Statistical Learning is the  {\it doubling constant} or the {\it local covering number}. For a collection of classifiers $\ff$, let us define the distance $\rho$ by $\rho(f, g) =\PP(f(x) \neq g(x))$. We denote by $\mathcal{M}(\ff, \eps)$ the {\it packing number} of $\ff$ with respect to  $\rho$:
$$\mathcal{M}(\ff, \eps):=\max_{\mathcal Q\subset \ff} \big\{|\mathcal Q|:\ \rho(f,g)\ge \eps\ \text{for any distinct}\ f,g\in \mathcal Q\big\}.$$
 Given any $f^*\in \mathcal{F}$, define the set $B(\mathcal{F}, f^*, \eps) := \big\{f \in \mathcal{F}:\ \rho(f, f^*) \le \eps\big\}$ of all classifiers from $\ff$ at distance at most $\eps$ from $f^*$. Finally, define the {\it doubling constant}
\begin{equation}
\label{dd}
D_{\eps}(P, \mathcal{F}, f^*) := \sup\limits_{\eps_{0} \ge \eps}\,\mathcal{M}(B(\mathcal{F}, f^*, 2\eps_0), \eps_0).
\end{equation}
We write $D_{\eps}$ instead of $D_{\eps}(P, \mathcal{F}, f^*)$ whenever the class of functions and $f^*$ are clear from the context. The logarithm of the doubling constant is referred to as the {\it doubling dimension}. It plays an important role in risk guaranties for some learning algorithms \cite{Bshouty09, Mendelson17, Zhivotovskiy16, Zhivotovskiy17}.

Let us reformulate \eqref{dd} in Computational Geometry terms of range spaces. Given a range space $(\X, \mathcal{R})$, we denote by $\mathcal{M}_{H}(\mathcal{R}, \eps)$ the maximal packing with respect to the distance $\rho$ defined by $\rho(R_1, R_2) =\PP(R_1\Delta R_2)$. Then the {\it doubling constant} is
\[
D_{\eps}(P, \mathcal{R}) = \sup\limits_{\eps_{0} \ge \eps}\,\mathcal{M}(\mathcal{R}_{\le 2\eps_{0}}, \eps_{0}).
\]

{ Section \ref{sec5} is devoted to various new bounds on the doubling constant. In particular, Lemma \ref{geombound} will give some sharp upper bounds for several interesting geometric range spaces.}

\subsection{A new bound for $\eps$-nets}

The first part of the following Theorem is an improvement of a recent theorem from \cite{MDG} (see also the related discussions there), and the technique of the proof is similar. We should note, however, that the technique of the proof is also closely related to the \emph{peeling technique} originating from the empirical processes theory and which is widely used in the Learning Theory \cite{Bartlett05, Bshouty09, Zhivotovskiy16, Zhivotovskiy17}. The second part of the Theorem complements the first one and has no known analogues.
\begin{thm}\label{thm4} Let $(\X, \RR)$ be a range space of VC-dimension $d$. Fix $\epsilon>0$ and let $D_{\epsilon}$ be an upper bound on the doubling constant of $(\X,\RR)$. Put $\tau_i:=\tau(2^i\eps)$ and $z:=1 + \lceil \log_2 \frac 1{\eps}\rceil$.

(i) If $D_{\epsilon} \ge 2\tau_1$, then there exists an $\eps$-net for $(\X,\RR)$ of size
\begin{equation}\label{eqmain}O\Big(\sum_{i=1}^{z}\Big( \log\frac{D_{\eps}}{\tau_i}+d\Big)\tau_i \Big).\end{equation}

(ii) If $D_{\epsilon}\le \frac 1{2\epsilon}$, there is an $\epsilon$-net for $(\X,\RR)$ of size $O\big(dD_{\epsilon}\log \frac {1}{\epsilon D_\epsilon}\big)$.
Moreover, for any $n, d$, $\epsilon>d/n$ and $D_{\epsilon}<\frac{1}{2\eps}$ there is a range space on $n$ points with uniform measure and $VC$ dimension at most $d$, doubling constant at most $D_{\epsilon}$ and the smallest $\epsilon$-net of size $\Omega\big(dD_{\epsilon}\log\frac {1}{\epsilon D_\epsilon}\big)$.
\end{thm}
{ \textbf{Remark.} We note that the additive form of the bound \eqref{eqmain} as well as of \eqref{eqmain2} is typical in Active Learning literature (see e.g., \cite{Zhang14}). In what follows we provide several straightforward relaxations of \eqref{eqmain} that have a non-additive form.}

\textbf{Remark. } Since $\tau_1\le \frac 1{2\eps}$, the two upper bounds from the theorem cover all possible range of values of $D_{\eps}$ except for $\tau_1\le D_{\eps} \le 2\tau_1$. If $D_{\eps}$ falls in that range then we may substitute $2\tau_1$ into the first bound instead of $D_{\eps}$ and get a valid bound.

The upper bound in Theorem~\ref{thm4} (i) is sharp for constant $d$, since it is a strengthening of an upper bound from \cite{MDG} (see Section~\ref{sec5}), which was recently shown to be tight in some specific cases in \cite{KMP17}.  The upper bound in Theorem~\ref{thm4} (ii) may be stated in terms of $\tau_i$, but the formulation gets rather complicated, so we decided to omit it.\vskip+0.1cm

Before proving the theorem, let us first obtain a handy corollary from \eqref{eqmain}. 
\begin{cor}\label{cormain} In the notation of Theorem~\ref{thm4}, assume that $D_\epsilon$ is an upper bound on the doubling constant of $(\X,\RR)$ and $\tau$ is its Alexander's capacity. If $D_{\epsilon}\ge e/\epsilon$, then there exists an $\eps$-net for $(\X,\RR)$ of size
\begin{equation}
\label{usefuleq}
O\Big(\frac{1}{\eps}\big(d+\log(\eps D_{\eps})\big)\Big).
\end{equation}
Similarly, if instead $D_{\epsilon} \ge 2\tau(\eps)$ there exists an $\eps$-net for $(\X,\RR)$ of size
\begin{equation}
\label{impr}
O\Big(\tau(\eps)\left(d+\log\frac{D_{\eps}}{\tau(\eps)}\right)\log \frac{1}{\eps}\Big).
\end{equation}
\end{cor}

{ \textbf{Example.} We argue that Corollary \ref{cormain} in this form is the most useful in various applications. In this example we focus on the range spaces with $D_{\eps} = O(\frac{1}{\eps})$. According to \eqref{usefuleq} for these range spaces there exist $\eps$-nets of size $O(\frac{d}{\eps})$, which is an improvement over Theorem \ref{thmvc} and is known to be optimal. It will be shown (Lemma \ref{geombound} in Section \ref{implications}) that the range spaces having $D_{\eps} = O(\frac{1}{\eps})$ include the set systems
\begin{itemize}
\item of VC dimension $d = 1$,
\item induced by half-spaces in $\mathbb{R}^2$ and $\mathbb{R}^3$,
\item induced by balls in $\mathbb{R}^2$,
\item induced by intervals in $\mathbb{R}$.
\end{itemize}
Although these geometric set systems are known to admit $\eps$-nets of size $O(\frac{d}{\eps})$
\footnote{In all mentioned cases it holds that $d \le 4$. Therefore, the bound $O(\frac{d}{\eps})$ is essentially $O(\frac{1}{\eps})$ in our case.} (we refer to the survey \cite{MV} for an extensive list of known upper bounds in geometric scenarious), Corollary \ref{cormain} highlights an explanation of this phenomenon: in order to show the existence of an $\eps$-net of size $O(\frac{d}{\eps})$ it is sufficient to show that $D_{\eps} = O(\frac{1}{\eps})$.

More importantly, our results will imply that for the geometric range spaces listed above, it holds that $D_{\eps} = O(\tau(\eps))$. According to \eqref{impr} combined with \eqref{usefuleq} implies the existence of $\eps$-nets of size
\begin{equation}
\label{optbound}
O\left(\min\left\{d\tau(\eps)\log \frac{1}{\eps}, \frac{d}{\eps}\right\}\right),
\end{equation}
which is again an improvement over both Corollary \ref{alexcap} and Corollary \ref{coralexnew}.
}

\begin{proof} (of Corollary \ref{cormain})
  The function $a\log \frac 1a$ is increasing for $a\in (0,1/e)$. Then, recalling that $\tau_i\le \frac 1{2^i\epsilon}$ and $\frac{\tau_i}{D_{\eps}}\le \frac {1}{D_{\eps}2^i\eps}\le \frac 1e$, we may apply it in the form $\frac{\tau_i}{D_{\eps}}\log\frac{D_{\eps}}{\tau_i}\le \frac{1}{2^i\eps D_{\eps}}\log(2^i\eps D_{\eps}).$ We get that
  \begin{align*}\sum_{i=1}^{z}\tau_i\log\frac{D_{\eps}}{\tau_i} & \le \frac 1{\eps}\sum_{i=1}^{z}2^{-i}\log(2^i\eps D_{\eps})
  \\
  &\le \frac{1}{\eps}\Big[\log(\eps D_{\eps})+\sum_{i=1}^{z}i2^{-i}\Big]=O\Big(\frac{1}{\eps}\log(\eps D_{\eps})\Big).
  \end{align*}
  Moreover, we have $\sum_{i=1}^z\tau_i\le \frac 1{\eps}$. We are left to substitute it into \eqref{eqmain}. The proof of the second bound is trivial given \eqref{eqmain} and the fact that $\tau(\eps)$ is nonincreasing.
  { We also remark that \eqref{eqmain} can be bounded by 
\[
O\left(\min\left\{\frac{d}{\epsilon}, 
   d\tau(\epsilon)\log\frac{1}{\epsilon}\right\} +
   \sum_{i=1}^{z}\tau_i\log\frac{D_\epsilon}{\tau_i}\right).
 \]
}
\end{proof}

Observe that the bound \eqref{impr} is always not worse than the best known bound of Corollary \ref{coralexnew} given in terms of Alexander's capacity alone. This fact follows directly from the bound \eqref{corol} below. We also note that Corollary 5 improves the best known upper bound for $\eps$-nets in terms of the doubling constant. Indeed, Theorem 8 in Bshouty et al. \cite{Bshouty09} says that Theorem \ref{thmvc} holds with
\[
m(\eps) = O\left(\frac{(d + \log D_{\eps})\sqrt{\log\frac{1}{\eps}}}{\eps} + \frac{\log\frac{1}{\delta}}{\eps}\right),
\]
{ which implies the existence of an $\eps$-net of size $O\left(\frac{(d + \log D_{\eps})\sqrt{\log\frac{1}{\eps}}}{\eps}\right)$.}

The following weaker bound, which is nevertheless stronger than Theorem~\ref{thmvc} and Theorem~\ref{Koltch}, follows from Theorem~\ref{thm4}. We sacrificed a factor in the logarithm in order to get a bound valid for any $D_{\eps}$.
\begin{cor}\label{cortri}
In the notation of Theorem~\ref{thm4}, there exists an $\eps$-net for $(\X,\RR)$ of size $$O\Big(\frac 1{\eps}\big(d+\log (D_{\eps})\big)\Big).$$
\end{cor}

\subsection{Proof of Theorem~\ref{thm4}.}
\subsubsection{(i). Upper bound}
We use the notation of the proof of Theorem~\ref{thmmain2}.  For each $i\ge 1$, by the definition of $D_{\eps}$, there is a maximal $\eps_{i-1}$-packing $\mathcal{Q}_{i}$ of size at most $D_{\eps}$ in $\mathcal{R}_{i}$. Note that, for each $R\in \RR_i$, there is $Q\in \mathcal Q_i$, such that $\PP(R\Delta Q)\le \eps_{i-1}$, which, together with $\PP(R),P(Q)\ge \eps_{i-1}$ implies $\PP(R\cap Q)\ge \eps_{i}/4\ge\PP(Q)/4.$ For each $Q\in \mathcal Q_i$, denote $\RR_i(Q):=\{R\in \RR_i:P(R\cap Q)\ge\PP(Q)/4\}$. { Since the sets of measure less than $\eps$ can be ignored, we are only interested in $\RR$ such that} $\RR\subseteq \cup_{i}\RR_i\subseteq \cup_{i}\cup_{Q\in \mathcal Q_i}\RR_i(Q)$. Therefore, a set, which would be a $1/4$-net (with respect to the conditional measure $\PP(\ |Q)$) for each of the families of ranges $\RR_i(Q)$, would be an $\eps$-net for $\RR$.

Recall that $\tau_i:= \tau(\eps_i)$. Thus, we have $\PP(\X^{(i)})\le \tau_i\eps_i$. Fix an absolute constant $c$ (defined later) and put $t_i=\log\frac{D_{\eps}}{\tau_i}$. Note that $1\le \tau_i/\tau_{i+1}\le 2$ and $D_{\eps}\ge 2\tau_i$ for any $i\ge 1$, therefore, $t_i\ge \log2$.

Consider a random sample $S_i$ of size $c(t_i+d)\tau_{i-1}$, sampled from $X$ according to the conditional distribution $\PP(\ |\X^{(i)})$.
Observe that we may think of $S_i\cap Q$ as a sample with elements distributed according to a conditional distribution $\PP(\ |Q)$.
We also have that for any $Q\in \mathcal Q_i$ it holds that $\PP(Q|\X^{(i)}) \ge \frac{\eps_{i-1}}{\tau_i\eps_i} = \frac{1}{2\tau_i}$. Using the Chernoff bound for an appropriately chosen $c$ and a fixed $Q$,  the event that at least $c(t_i+d)/8$ of the elements of $S_i$ belong to $Q$ has probability at least $1-\exp(-t_i - \log 2)$.
By Theorem~\ref{thmvc} and for an appropriate value of  $c$, the set $S_i\cap Q$ is a $1/4$-net (with respect to a conditional measure $\PP(\ |Q)$) for $\RR_i(Q)$ with probability at least $1-\exp(-t_i-\log 2)$ { given that at least $c(t_i+d)/8$ of the elements of $S_i$ belong to $Q$}. Using a union bound, we conclude that, for a fixed $Q$, the set $S_i\cap Q$ is a $1/4$-net for $\RR_i(Q)$ (with respect to $\PP(\ |Q)$) with probability at least $1-\exp(-t_i)$.

Put $S:=\cup_{i=1}^z S_i$. Therefore, the expectation of the number $M_i$ of $Q\in \mathcal Q_i$, for which $S$ is not a $1/4$-net, is
$$\E M_i  \le D_{\eps}\exp(-t_i) = \tau_i.$$

On the other hand, the size $N$ of $S$ is

$$N \le \sum_{i=1}^{z} c(t_i+d)\tau_{i-1}.$$

Using the Markov inequality, with positive probability, it holds that $\sum_{i=1}^{z} M_i\le 3\sum_{i=1}^{z}\tau_i$. We fix such a set $S$. Next, we manually find a $1/4$-net (with respect to conditional measure $\PP(\ |Q)$) for each of the $\RR_i(Q)$ that contribute to $M_i$. Using Theorem~\ref{thmvc} for $\eps = \frac{1}{4}$, we conclude that we need to add a set $A$ of additional $O(d\sum_{i=1}^{z} \tau_i)$ points to the $\eps$-net in order to cover the remaining sets that might be still uncovered. Therefore, in total we get an $\eps$-net of size
$$O\Big(\sum_{i=1}^{z} (t_i+d)\tau_i\Big).$$

\qed

\subsubsection{(ii). Upper bound}
We work in the notation of Theorem~\ref{thmmain2} and (i). Put  $\lceil \log_2 \frac e{D_{\eps}\eps}\rceil=:i_0$. Then all ranges in $\RR':=\RR\setminus \cup_{i=1}^{i_0} \RR_i$ have measure at least $\frac{e}{D_{\eps}} \ge 2e\eps$. We know that the doubling constant of the range space $(\X,\RR')$ is bounded by $D_{\eps}$, and, applying \eqref{usefuleq} of Corollary~\ref{cormain} {(which is possible, since the proof of Corollary~\ref{cormain} uses only item (i) of Theorem \ref{thm4})} with $\epsilon$ equal to $\frac e{D_{\eps}}$, we conclude that there is an $\epsilon$-net for $(\X,\RR')$ of size at most $O(dD_{\eps})$. Therefore, to conclude the proof of this part of the theorem, it is sufficient to show that for each $i=1,\ldots, i_0$ the range space $\RR_i$ has $\eps$-net of size $O(dD_{\epsilon})$.

Consider $\mathcal Q_i$ and the corresponding $\RR_i(Q)$ for $Q\in \mathcal Q_i$. Then for a fixed $i$ we have $|\mathcal Q_i|\le D_{\eps}$ and for each $\RR_i(Q)$ there is a $1/4$-net (with respect to $\PP(\ |Q)$) of size $O(d)$. Thus, there is an $\eps$-net for $\RR_i$ of size $O(dD_{\epsilon})$. The total size of the $\eps$-net is $O(di_0D_{\eps})$.

\subsubsection{(ii). Lower bound}
To construct the lower bound, we consider the finite set $\X$ of $n$ elements equipped with a uniform measure.
For simplicity, let us assume that $\eps n= k$ is an integer number. For each $i$ fix $X'^{(i)}$ of cardinality $k2^i+d-1$ and consider the following collection of ranges: $\RR'_i:=\{R\subset X'^{(i)}: |R|=k2^i\}$. Next, form a range space $(\X^{(i)},\RR_i)$ by taking $l$ disjoint copies of $\RR'_i$ on disjoint sets $X'^{(i)}$. Finally, define $(\X,\RR)$ to be the union of disjoint copies of $(\X^{(i)},\RR_i)$ for $i=0,\ldots, m-1$. Again, for simplicity we assume that $n= \sum_{i=0}^{m-1}|\X^{(i)}|=(d-1)lm+lk\sum_{i=0}^{m-1}2^i$. Knowing that $d$ is not too large, we get that $lk2^{m}<n< lk2^{m+1}$.

It is clear that the VC-dimension of $(\X,\RR)$ is determined by each of the range spaces $(\X'^{(i)},\RR'_i)$ and is equal to $d$. Next, the smallest $\eps$-net for $(\X,\RR)$ has size at least $lm$ times the smallest $\eps$-net for each $(\X'^{(i)},\RR'_i)$, which gives us $lmd$. Let us calculate the doubling constant of $(\X,\RR)$. For any $\gamma\ge \eps$, $\gamma n\in\N$, choose $j:=\min\{i:2^ik>\gamma n\}$. How large can a packing of balls of radius $\gamma$ be in $(\X,\RR_{\le 2\gamma})$? We should include in the packing exactly one set from each $\RR_i'$ for $i=j$ and $j-1$, which gives $2l$ balls. All the sets from $\RR_i'$ for $i\le j-2$ will be covered by one ball of radius $\gamma$ with the center in any of those sets, and the sets from $\RR_i'$ for $i\ge j+1$ are bigger than $2\gamma$ and are not present in the family. Therefore, $D_{\eps}\le 2l+1$ (actually, $D_{\eps}=2l+1$).

We have that $n=O(D_{\eps} k 2^{m})$ and $\eps =\frac kn=\Omega( \frac{1}{D_{\eps}2^m})$, which means that $\log \frac 1{D_{\eps}\eps}=O(m)$. Therefore, the minimum size of an $\eps$-net is $lmd= \Omega \big(d D_{\eps}\log \frac 1{D_{\eps}\eps}\big)$.
\qed

\textbf{Remark. } The family that provides the lower bound above  may be used to show that Theorem~\ref{thmmain2} is tight at least for constant $\tau_i$. Putting $l=1$ in the construction $(\X,\RR)$ above, we get that $D_{\eps}$ is a constant, $\tau(\eps)<2+\frac d{\eps n}$, and that the minimum size of an $\eps$-net is $\Omega \big(d\log \frac 1{\eps}\big)$.

It is likely that we may even show that the bound of Theorem~\ref{thmmain2} is tight for slowly growing $\tau_i$ (that is, that the factor $\log \tau_i$ is also necessary) by replacing $\RR_i$ with disjoint copies of families that provide lower bound in Theorem~\ref{thmvc}.

\section{Packing numbers for VC classes}
\label{sec3.5}
In this section we discuss several packing results for VC classes of functions, which would be useful in getting upper bounds on the doubling constant.
At first we recall the following classical result due to Haussler. As before, for a pair of binary functions define $\rho(f, g) =\PP(f(x) \neq g(x))$. Note that any result for a class of binary functions is directly translated to range spaces.
\begin{thmx}(Haussler \cite{Haussler95})
\label{haussler}
Consider a class $\mathcal{F}$ of binary functions of VC dimension at most $d$, such that for any distinct $f, g \in \mathcal{F}$ we have $\rho(f, g) \ge \eps$. Then
\begin{equation}
\label{hauss}
|\mathcal{F}| \le e(d + 1)\left(\frac{2e}{\eps}\right)^d.
\end{equation}
\end{thmx}
The next lemma directly follows from the result of Chazelle (this fact was observed in \cite{MDG}).
\begin{lem}[Chazelle \cite{Chaz92}]
\label{chaz}
Consider a class $\mathcal{F}$ of binary functions of VC dimension at most $d$ with  $\rho(f, g) \ge \eps$ for any distinct $f, g \in \mathcal{F}$. If the measure $\PP$ is uniform on a finite set, then we have
\[
|\mathcal F| \le 2\E|\mathcal F|{}_A|,
\]
where $A$ is an i.i.d. sample of size $n = \frac{4d}{\eps}$ from $\X$ sampled according to $\PP$ and $\mathcal F|{}_A$ denotes the set of projections of $\mathcal{F}$ on the sample $A$.

\end{lem}
This lemma directly implies the version of Haussler's original lemma for the uniform distribution. Using the Vapnik-Chervonenkis-Sauer bound \eqref{vcs} we have $|\mathcal F| \le 2\sum\limits_{i = 0}^{d}{n \choose i} \le 2\left(\frac{en}{d}\right)^d \le 2\left(\frac{4e}{\eps}\right)^d$. However,  constants in this deduction are somewhat worse than in \eqref{hauss}. In what follows, we give a more general result with a short proof that directly implies the Haussler's bound. As opposed to Lemma \ref{chaz}, our proof will be based on a purely statistical approach. In fact, the bound on the packing number will be derived as a byproduct of the minimax analysis of the learning rates of the so-called one-inclusion graph algorithm. The analysis is inspired by the minimax lower bounds provided in \cite{Benedek91}.

\begin{lem}
\label{hausslerlem}
Fix any $\delta \in (0, 1)$. Consider a class $\mathcal{F}$ of binary function of VC dimension at most $d$ such that for any distinct $f, g \in \mathcal{F}$ it holds that $\rho(f, g) \ge \eps$. Then
\[
|\mathcal F| \le \frac{\E|\mathcal F|{}_A|}{1 - \delta},
\]
where $A$ is an i.i.d. sample of size $n = \frac{2d}{\eps\delta}$ from $\X$ sampled according to $\PP$.
\end{lem}

This Lemma will be used below in the context of bounds which depend on shallow cell complexity (see Section~\ref{sec5}). To prove this bound, we need the following result from Learning Theory. Note that the proof of the next Lemma is not based on the bound on packing numbers. The discussion of the one-inclusion graph algorithm is moved to Appendix.
\begin{lem}
\label{oig}
In the realizable case of classification there is a deterministic learning algorithm such that, given an i.i.d sample $\bigl((x_i, f^*(x_i))\bigr)_{i = 1}^n$ of size $n = \frac{2d}{\eps\delta}$, it produces a classifier $\hat{f}$ with $\rho(\hat{f}, f^*) < \eps/2$ with probability at least $1 - \delta$ over the learning sample.
\end{lem}
\begin{proof}
It follows from the fact that there is a strategy (namely \emph{one-inclusion graph algorithm} \cite{HLW}) with an expected error $\E \rho(\hat{f}, f^{*}) \le \frac{d}{n + 1} < \frac{d}{n}$, where expectation is taken with respect to the i.i.d random sample $\bigl((x_i, f^*(x_i))\bigr)_{i = 1}^n$ for an arbitrary target function $f^* \in \mathcal{F}$. { We define the algorithm and sketch the proof of the risk bound in Appendix \ref{oig}}. Using Markov inequality we have $\PP(\rho(\hat{f}, f^{*}) \ge \eps/2) \le \frac{2\E \rho(\hat{f},  f^{*})}{\eps} < \frac{2d}{n\eps}$. We fix $n = \frac{2d}{\eps\delta}$ and get that with probability at least $1 - \delta$, it holds that $\rho(\hat{f}, f^{*}) < \eps/2$.
\end{proof}
\medskip
\begin{proof}{(of Lemma \ref{hausslerlem}.)}
For $n = \frac{2d}{\eps\delta}$ denote the output of the learning algorithm of Lemma \ref{oig} based on the sample $\bigl((x_i, f(x_i))\bigr)_{i = 1}^n$ by $\hat{g}_f$. Define the uniform measure $\pi$ on $\mathcal{F}$. Due to Lemma~\ref{oig} we have
 $
\E_{f \sim \pi}P(\rho(\hat{g}_f, f) < \eps/2) \ge 1 - \delta.
$
Assume that for a pair of distinct $f, h \in \mathcal{F}$ it holds that $f(x_i) = h(x_i)$ for $i = 1,\ldots, n$, i.e., they have the same projection on the sample. Since our prediction strategy is deterministic we have $\hat{g}_f = \hat{g}_h$. However, it is not possible that simultaneously we have $\rho(\hat{g}_f, f) < \eps/2$ and $\rho(\hat{g}_h, h) < \eps/2$ since in this case by the triangle inequality $\rho(f, h) \le \rho(\hat{g}_f, f) + \rho(\hat{g}_h, h) < \eps$, but at the same time from the statement of the Lemma we have $\rho(f, h) \ge \eps$. Thus taking into account that for each random sample $A$ there are at most $|\ff|_A|$ different functions $\hat g_f$ that may serve as an output, and each corresponds to at most one function from $\ff$, we get
\[
1 - \delta \le \E_{f \sim \pi}P\left(\rho(\hat{g}_f, f) \le \eps/2\right) = \frac{1}{|\mathcal{F}|}\E\sum\limits_{f \in \mathcal{F}}\ind[\rho(\hat{g}_f, f) < \eps/2] \le \frac{1}{|\mathcal{F}|}\E|\mathcal F|{}_A|.
\]
\end{proof}
Taking $\delta = \frac{1}{2}$ the uniform measure as $\PP$ in Lemma \ref{hausslerlem}, one recovers Lemma \ref{chaz}.
To recover the result of Haussler \eqref{hauss} we use the Vapnik-Chervonenkis-Sauer bound \eqref{vcs} again together with Lemma \ref{hausslerlem}:
\begin{equation}
\label{Haussineq}
|\mathcal{F}| \le \inf\limits_{\delta \in (0, 1)}\frac{\E|\mathcal{F}|{}_{A}|}{1 - \delta} \le \inf\limits_{\delta \in (0, 1)}\frac{1}{1 - \delta}{\left(\frac{2e}{\eps\delta}\right)^d} \le e(d + 1)\left(\frac{2e}{\eps}\right)^{d},
\end{equation}
{ which is obtained by choosing $\delta = \frac{d}{d + 1}$.}

\textbf{Remark.}
{ If instead of Lemma \ref{oig} we use a bound $O(\frac{d\log(\frac{1}{\eps})}{\eps})$ for consistent learning algorithms (Theorem \ref{vcbound}), we obtain the weaker bound}
\[
|\mathcal{F}| = O\left(\left(\frac{1}{\eps}\log \frac{1}{\eps}\right)^d\right),
\] where constants in $O$ depend on $d$, which coincides with the original bound of Dudley \cite{Dudley82}. In fact, using our technique any deterministic learning algorithm with provable guarantees on probability of misclassification will provide an upper bound on packing numbers. For example, we may replace the algorithm in Lemma \ref{oig} by the recent result \cite{Han15}. In this case the bounds will be the same as in Theorem \ref{haussler} up to absolute constants.

\section{Upper bounds on the doubling constant}\label{sec5}

To motivate some of the results we prove below, we mention the following bound on $D_{\eps}$.\begin{thmx}[\cite{Hanneke15}, Theorem 17]\label{thmha} We have $D_{\eps} \le (c\tau(\eps))^{c_1d}$ for some absolute constants $c, c_1>1$.
\end{thmx}
In this form it is usually sufficient for statistical applications,  but in what follows we shall need a bound, tight in terms of $c_1$. One of the corollaries of the results proved below is that for a range space $(\X, \mathcal{R})$ of VC dimension $d$ it holds that $D_{\eps} \le (c\tau(\eps))^d$ for some $c > 0$ (cf. Corollary~\ref{corol1}).

Let us start with defining another important measure of complexity, called the \emph{shallow-cell complexity}. It was introduced recently \cite{Aro, Chan, KMP17, Var} for a  more refined analysis of the projections of the range spaces, than the one that we can extract from $\pi_{\RR}(y)$ and the VC-dimension. For the relation of shallow-cell complexity with the so-called \emph{union complexity}, see \cite{KMP17}. Here we give a definition that slightly differs from the one given in \cite{MDG}, \cite{KMP17}: we do not isolate the term $|Y|$ in the projections on $Y$, but rather include it into the shallow cell complexity function.

 A range space $(\X, \mathcal{R})$ has {\it shallow-cell complexity} $\varphi : \mathbb{N} \times\mathbb{N} \to \mathbb{N}$ if for every $Y \subseteq X$, the number of sets of size at most $\ell$ in the system $\mathcal{R}|_{Y}$ is at most $\varphi(|Y|, \ell)$. In all known geometric applications it is sufficient to consider the functions of the form $\varphi(|Y|,\ell)=\varphi'(|Y|)\ell^{c_\mathcal{R}}$ for some constant ${c_\mathcal{R}}$, and, if this is the case for a range space, then  we say that the range space {\it has shallow cell complexity $(\varphi', c_\mathcal{R})$}. Thus, the difference with the projection function is that $\varphi$ bounds the number of sets of different sizes separately. We should note that for many known geometric scenarios very tight bounds on the shallow cell complexity are known (we refer to the survey \cite{MV} for a complete list).

We provide two upper bounds on the doubling constant in terms of shallow-cell complexity and Alexander's capacity. The proof of the forthcoming results are postponed to Appendix to facilitate the reading.
\begin{lem}\label{lembd1}
Assume that the range space $(\X, \mathcal{R})$ has a shallow-cell complexity $(\varphi', c_{\mathcal{R}})$ such that $\varphi'(x) \le c_1x^{k}$ for some $c_1, k > 0$ and Alexander's capacity $\tau$. Then
\[
D_\eps \le C(k,c_{\RR})\tau^{k}(\eps)\log^{k + c_{\mathcal{R}}}\tau(\eps),
\]
where $C(k,c_{\RR}) = O\big((c(k + c_{\mathcal{R}})\log(k + c_{\mathcal{R}}))^{k + c_{\mathcal{R}}}\big)$ for some $c>0$.
\end{lem}

The next lemma is better than the previous one in many cases, but depends explicitly on VC dimension. { We remark that Lemma \ref{lembd1} involves only the shatter function, but not  VC dimension.}
\begin{lem}
\label{shallow}
Assume that the range space $(\X, \mathcal{R})$ has VC-dimension $d$, shallow cell complexity $\varphi$ and Alexander's capacity $\tau$. Then
\[
D_{\eps} \le 6\varphi(8d\tau(\eps), 24d).
\]
\end{lem}
We immediately have the following result, already mentioned in Section~\ref{sec3}.
\begin{cor}\label{corol1}
For a range space $(\X, \mathcal{R})$ of VC-dimension $d$ there exists $c>0$ such that
\begin{equation}
\label{corol}
D_{\eps} \le \big(c\tau(\eps)\big)^d.
\end{equation}
\end{cor}
\begin{proof}
{ The result follows since the shallow-cell complexity can be upper bounded by the Vapnik-Chervonenkis-Sauer-Shelah inequality.}
\end{proof}
{
Finally, we provide Lemma that proves the bound \eqref{optbound} in some geometric scenarios.
\begin{lem}
\label{geombound}
For the range spaces induced by half-spaces in $\mathbb{R}^2$ and $\mathbb{R}^3$, by balls in $\mathbb{R}^2$, by intervals in $\mathbb{R}$ and for range spaces of VC dimension $d = 1$ it holds that
\begin{equation}
\label{taueps}
D_{\eps} = O(\tau(\eps)).
\end{equation}
\end{lem}
\begin{proof}
The result for $d = 1$ follows immediately from Corollary \ref{corol1}. In view of this result it is worth noting the importance of the inequality \eqref{corol} compared to Theorem \eqref{thmha}. The latter will not allow us to prove \eqref{taueps}.

The proof of the remaining part follows from Lemma \ref{shallow}. Indeed, for range spaces induced by half-spaces in $\mathbb{R}^2$ and $\mathbb{R}^3$, by balls in $\mathbb{R}^2$, by intervals in $\mathbb{R}$ the known bound on the shallow-cell complexity (see Table 47.1.1 in \cite{MV}; in all mentioned cases the shallow cell complexity $\varphi(x,\ell) = O(x\ell^{C})$, where $C$ is a constant that does not depend on $x$) together with the fact that $d \le 4$ imply that $\varphi(8d\tau(\eps), 24d) = O(\tau(\eps))$. The claim follows.
\end{proof}
}
{
\section{Concluding remarks and discussions}
\label{implications}\label{sec6}\label{sec7}
Since our bounds depend on the Alexander's capacity $\tau$ we need to recall the following related quantity. Define the \emph{star number} \cite{Hanneke15} as the largest integer such that there exist distinct
points $x_1, . . . , x_{\mathfrak{s}} \in \X$ and classifiers $f_0, f_1, . . . , f_{\mathfrak{s}} \in \mathcal F$ with the property that for all $i$,
\[
\{x \in X: f_0(x) \neq f_i(x)\} \cap \{x_1, \ldots, x_{\mathfrak{s}}\} = \{x_i\}.
\]
Theorem 10 in \cite{Hanneke15} shows that 
\begin{equation}
\label{starnum}
\sup\limits_{\PP, f^* \in \mathcal F}\tau(\eps) = \min\left\{\mathfrak{s}, \frac{1}{\eps}\right\},
\end{equation}
and, in particular, we always have $\tau(\eps) \le \mathfrak{s}$ and our upper bounds can be changed accordingly. The star number is known to characterize the learning rates in Active Learning \cite{Hanneke15} and has important relations to Teaching as discussed below. However, in the world of $\epsilon$-nets the range spaces induced by classes with the finite star number are trivial: these are the range spaces that have finite hitting sets. We recall that set $H \subseteq \X$ is a hitting set of $(\X, \RR)$ if it intersects each set in $\RR$. 
\begin{obs}
\label{obsepsnet}
For any VC class $\mathcal F$ with the finite star number $\mathfrak{s}$ it holds that for any $f^* \in F$ the range space
\[
\big\{\{x \in \X: f(x)\ne f^*(x)\}: f\in\ff\big\},
\]
has a finite hitting set of size $\mathfrak{s}$. In particular, this means that there is $\eps$-net of size $\mathfrak{s}$ for any $\eps > 0$ and all $f \in \FF$ have finite teaching sets of size $\mathfrak{s}$.
\end{obs}

\begin{proof}

By Theorem 13 in \cite{Hanneke15} if $\mathfrak{s} < \infty$ then it is enough for the teacher to show at most $\mathfrak{s}$ points to allow the learner to reconstruct $f^*$ on a finite set $\mathrm{supp}(\PP)$ by any consistent learning procedure. This implies the remaining claims in a straightforward way.

\end{proof}

Therefore, in view of Observation \ref{obsepsnet} the interesting case for $\eps$-nets is where $\mathfrak{s} = \infty$, but $\tau(\eps) = o\left(\frac{1}{\eps}\right)$. We note that this scenario was considered in the literature in the context of $\eps$-nets. In \cite{Balcan13} authors prove and discuss the existence of $\eps$-nets of size $O(\frac{d}{\eps})$ (which is an improvement over the standard $O(\frac{d \log \frac{1}{\eps}}{\eps})$ bound) for a set of regions of disagreement between all possible linear classifiers passing through the origin in $\mathbb{R}^d$ and the linear fixed classifier, when the distribution $\PP$ is zero-mean, isotropic and log-concave. It is easy to see that in this case for $d \ge 2$ it holds that $\mathfrak{s} = \infty$. Their bound is based on the improved version of Theorem \ref{vcbound}. However, it is not difficult to see that at least for some particular distributions (like uniform distribution on the unit sphere) even finite $\eps$-nets exist, that are the hitting sets. More specifically, our bound \eqref{eqtau} (given the fact the Alexander's capacity is bounded in this case \cite{Hanneke14}) gives the result that scales as $O(\log \frac{1}{\eps})$, which is significantly better than $O(\frac{d}{\eps})$ claimed in \cite{Balcan13}.

We know \eqref{starnum} that boundedness of the star number is a necessary and sufficient condition for boundedness of $\tau(\eps)$ for all distributions, and it also implies the existence of finite hitting sets for corresponding range spaces. Our results answer the following natural question. Is it true that boundedness of $\tau(\eps)$ for a particular distribution implies the existence of $\eps$-nets with their sizes not depending on $\eps$ ?  We show that it is not true and in general, it is not possible to get rid of the factor $\log \frac{1}{\eps}$ in the bound \eqref{impr} since there are range spaces with $\tau(\eps) = O(1)$ and with the smallest $\eps$-net of size $\Omega(\log \frac{1}{\eps})$. See the remark after the proof of Theorem~\ref{thm4}.  

In general, the $\epsilon$-net theorems in this paper are arranged from the weaker to the stronger ones. Below we only discuss the strength of the bounds given, and mostly avoid discussing the algorithms. We focus only on the implications of certain results on the existence of $\eps$-nets. 

Theorems~\ref{thmvc},~\ref{thmvcrev},~\ref{vcbound} are weaker than any other result given in the paper. Theorem~\ref{Koltch} is stronger than the previous ones, and its bound is implied for relatively small $\tau(\epsilon)$ by Corollary~\ref{alexcap}.

Theorem~\ref{thmmain2} implies both Theorem~\ref{Koltch} and Corollary~\ref{alexcap}. Indeed, Theorem~\ref{Koltch} follows easily from the fact that $\tau_i\le \frac{1}{2^i\eps}$, and thus $\sum_i\tau_i\le \frac 1{\epsilon}$, and Corollary~\ref{alexcap} follows from Corollary~\ref{coralexnew}. We also note that the bounds in Theorem~\ref{thmmain2} are strictly stronger in many cases.

Speaking of the bounds that make use of the doubling constant, they are stronger than all the previous. In particular, even Corollary~\ref{cortri} together with a weak bound on the doubling dimension, given in Theorem~\ref{thmha}, implies many of the previous bounds (except for Corollary~\ref{alexcap} and Theorem~\ref{thmmain2}), giving the bound $O\big(\frac{d}{\eps}\log\tau(\eps)\big)$, and Corollary~\ref{cormain} combined with Theorem~\ref{thmha}, implies Corollaries~\ref{alexcap} and~\ref{coralexnew}.

It is also easy to see that Theorem~\ref{thm4} combined with Theorem~\ref{thmha} implies Theorem~\ref{thmmain2}. But its full strength becomes clear when the doubling constant is relatively small. Then the fact that we divide $D_{\eps}$ by $\tau(\eps)$ in the logarithm may play a crucial role since it allows us to get rid of the logarithm in some cases. In this context, sharper bounds on the doubling constant that make use of the shallow-cell complexity come into play as shown above.

In the context of our bounds on the doubling constant it is worth mentioning the following Lemma.

{ \begin{lemm}[Shallow-Packing Lemma \cite{Dutta15, Mustafa16}]\label{lemmdg}
Let $\X$ be an $n$-element set equipped with the uniform measure $\PP$. Assume that a range space $(\X,\mathcal R)$ of at most $k$-element sets has VC dimension $d$, shallow-cell complexity $\varphi$. If, for any distinct $R_1,R_2\in \mathcal R$, we have $\PP(R_1\Delta R_2)\ge \gamma$ then
\[
|\mathcal{R}| = O\left(\varphi\left(\frac{4d}{\gamma},\frac {12kd}{n\gamma}\right)\right).
\]
\end{lemm}
 }

The previous best bound on the $\epsilon$-nets, which used the notion of shallow-cell complexity, was as follows (see \cite{Chan} and a simplified proof \cite{MDG}).

\begin{thmx}[\cite{Chan, MDG}]\label{thmmdg} Let $(\X,\RR)$, $|\X|=n$ be a range space with uniform distribution of VC-dimension $d$ and shallow-cell complexity $\varphi(\cdot,\cdot)$, where $\varphi(\cdot,\cdot)$ is a non-decreasing function in the first variable. Then there exists an $\eps$-net for $\RR$ of size $O\big(\frac 1{\eps}(\log\big(\eps\varphi(\frac {8d}{\eps}, 24d)\big)+d)\big)$. In particular, if $\RR$ has shallow-cell complexity $(\varphi',c_{R})$ and finite VC-dimension, then there exists an $\eps$-net for $\RR$ of size $O\big(\frac 1{\eps}\log (\eps\varphi^{\prime}\big(\frac{1}{\eps}\big))\big)$.
\end{thmx}

Our bound \eqref{eqmain} of Theorem \ref{thm4} recovers Theorem \ref{thmmdg} if we upper bound $\tau(\eps)$ by $1/\eps$ and $D_{\eps}$ by the upper bound of Lemma \ref{shallow}.
}
Overall, we feel that the doubling constant is the right general parameter to look at for $\eps$-nets. From this perspective, the notions like Alexander's capacity and the shallow-cell complexity are simply the ways to control the doubling constant. The doubling constant together with Alexander's capacity control almost all possible ranges of sizes of $\eps$-nets. Moreover, the extensions for the quantities like the doubling constant to the non-binary cases are straightforward (see \cite{Mendelson17} for these extensions related to the Learning Theory), while the notion of the shallow-cell complexity is very specific to set systems.

\section{Appendix}
\subsection{CAL algorithm for $\eps$-nets}
\label{cal}
The formal procedure behind the CAL algorithm, adapted to the case of range spaces is written below. The idea behind the algorithm applied to the range spaces is very simple and natural. We sample random points according to $\PP$ and add them to the $\eps$-net one by one. But contrary to the strategy of Theorem \ref{thmvc} the new point is added iff it is contained in a least one of the sets which were not hit (by hitting we mean having a nonempty intersection) by the points which were added to the $\eps$-net on previous steps. In some sense, we never add a new point if it does not hit any set that was not hit before.
\begin{algorithm}
\caption{CAL for $\eps$-nets}
\begin{algorithmic}[1]
\Procedure{}{}
\State $m \gets 0, t \gets 0, \mathcal{R}_0 \gets \mathcal{R}$
\While {$t < n\ \text{and}\  m < 2^n$}
\State $m \gets m + 1$
\If {$x_m \in \cup_{R \in \mathcal{R}_{m-1}} R$}
\State Add $x_m$ to $\eps$-net,
\State $\mathcal{R}_m \gets \{R \in \mathcal{R}_{m - 1}| R \cap x_m = \emptyset\}$,
\State $t \gets t + 1$.
\Else
\State $\mathcal{R}_m \gets \mathcal{R}_{m - 1}$.
\EndIf
\EndWhile
\EndProcedure
\end{algorithmic}
\end{algorithm}

Corollary \ref{alexcap} follows in a straightforward manner. Indeed for $\delta = \frac{1}{2}$ CAL returns the desired $\eps$-net with probability at least $\frac{1}{2}$. This implies the existence of the $\eps$-net of size $O\Big(\tau(\eps)\big[d\log\tau(\eps) + \log\log\frac 1{\eps}\big]\log\frac 1{\eps}\Big)$.

\subsection{Proofs of Section \ref{sec5}}
\begin{proof}{(Proof of Lemma~\ref{lembd1})} We follow the classic strategy of \cite{Dudley82} with some necessary modifications.
 In what follows we assume $D_{\eps} > 10$. Without loss of generality, assume that the supremum in the definition of $D_\eps$ is achieved at $\eps$.
Denote the corresponding maximal packing by $\mathcal{Q}$, where $|\mathcal{Q}| = D_\eps$. We have $\PP(Q_1\Delta Q_2) \ge \eps$ for any two $Q_{1}, Q_2 \in \mathcal{Q}$  and $\PP(Q) \le 2\eps$ for any $Q \in \mathcal Q$. By the definition of $\tau$ we have $$\PP(\X') \le 2\tau(\eps)\eps\ \ \ \ \ \ \text{for} \ \ \ \ \X':=\mathrm{supp}\, \mathcal Q\subset \cup_{R \in \mathcal{R}_{\le 2\eps}} R,$$
{ where we used $\tau(\eps) \ge \tau(2\eps)$.}
Consider the conditional distribution $\PP(\ |\X')$ and denote it by $\PP'$. Note that $\PP'(Q_1\Delta Q_2)>\frac{\eps}{\PP(\X')}\ge \frac 1{2\tau(\eps)}$ for any distinct $Q_1,Q_2\in \mathcal Q$. Note that in what follows we work with $\PP'$ only. In particular, all expectations below are computed w.r.t. $\PP'$.

 Take a random i.i.d. sample $A$ of size $m$ according to  $\PP'$, where $m := \frac{2\log D_\eps \,\PP(\X')}{\eps}$. Note that $$m\le 4\tau(\eps)\log D_{\eps}.$$
Given any two $Q_{1}, Q_2 \in \mathcal{Q}$ with this property we have that $(Q_{1}\Delta Q_2)\cap A \neq \emptyset$ with probability greater than
$$1 - \Big(1 - \frac{\eps}{\PP(\X')}\Big)^m > 1 - \exp\Big(-\frac{\eps m}{\PP(\X')}\Big) > 1 - \frac{1}{D_{\eps}^2}.$$

Using a union bound and summing over all unordered pairs in the packing we conclude that with probability strictly bigger than $\frac{1}{2}$
$$\text{each set in } \mathcal Q \text{ has a unique projection on the random sample } A.$$

At the same time for any $Q \in \mathcal{Q}$ it holds $\E |A \cap Q| \le m \PP'(Q) \le 4\log D_{\eps}$, since $\PP'(Q) \le \frac{2\eps}{\PP(\X')}$. We note that $|A\cap Q|$ is upper bounded by
a random variable which counts the number of elements of $A$ which belong to $Q$. The latter random variable has a binomial distribution and using the Chernoff bound together with the union bound we have $$\PP\big[\nexists Q \in \mathcal{Q}: |Q\cap A| \ge 8\log(D_\eps)\big] \ge1- D_\eps\exp\Big(-\frac{4\log(D_\eps)}{3}\Big) = 1-\frac{1}{D_{\eps}^{1/3}} \ge \frac 1 2.$$

Using a union bound, we conclude that both displayed events hold with positive probability simultaneously.
By the definition of shallow-packing we have for some absolute constant $C$
\[
D_\eps = |\mathcal{Q}|_{A}| \le \varphi'(m)(8\log D_{\eps})^{c_{\mathcal{R}}} \le C\tau^k(\eps)(8\log D_\eps)^{k + c_{\mathcal{R}}}.
\]
It is straightforward to check that the last inequality implies
\[
D_\eps \le C'\tau^k(\eps)\big(c(k + c_{\mathcal{R}})\log(k + c_{\mathcal{R}})\log \tau(\eps)\big)^{k + c_{\mathcal{R}}}.
\]
for some absolute $C', c > 0$.
\end{proof}
\vskip+0.2cm

\begin{proof}(Proof of Lemma \ref{shallow}) The proof follows the same logic as the previous one (the simplified technique we are following appeared first in \cite{Mustafa16}).   Without the loss of generality we assume that the supremum in the definition of doubling constant is achieved at $\eps_0 = \eps$. In particular, it means that the largest packing should contain sets of probability {measure not greater than $2\eps$}. We work in the setting of the proof of Lemma~\ref{lembd1} w.r.t. $\mathcal Q$ and $\PP'$.

Applying Lemma~\ref{hausslerlem} for $\PP'$ and $\delta = \frac{1}{2}$ (note that we may optimize with respect to $\delta$ and improve constant factors), we conclude that for an i.i.d. sample $A$ of size $$n:=\frac{4d\PP(\X')}{\eps}\le 8d\tau(\eps),$$ from $\PP'$ we have $|\mathcal Q|\le 2\E|\mathcal Q|_A|$.

At the same time for any $Q \in \mathcal{Q}$ it holds that $\E |A \cap Q| \le n \PP'(Q) \le 8d$, since $\PP'(Q) \le \frac{2\eps}{\PP(\X')}$.
Consider $\mathcal{Q}_1 := \{Q \in \mathcal{Q}:\ |A \cap Q| > 24d\}$. Using Markov's inequality we have for any $Q \in \mathcal{Q}$ that $$\PP[Q \in \mathcal{Q}_1] = \PP[|A \cap Q| > 24d] \le \frac{8d}{24d} = \frac{1}{3}.$$  Finally, we have
\begin{align*}
|\mathcal Q| &\le 2\E|\mathcal Q|{}_A|
\le 2\E|\mathcal Q_1| + 2\E|(\mathcal Q \setminus \mathcal Q_1)|{}_{A}|
\le \frac{2|\mathcal Q|}{3} + 2\E|(\mathcal Q \setminus \mathcal Q_1)|{}_{A}|.
\end{align*}
Rearranging, we obtain $|\mathcal Q|\le 6 \E|(\mathcal Q \setminus \mathcal Q_1)|{}_{A}| \le 6\varphi(8d\tau(\eps), 24d)$.
\end{proof}
{
\subsection{Sketch of the proof of one-inclusion graph algorithm error bound}
\label{oig}
Given class of binary functions $\FF$ and a finite sample $S \subset \X$ we consider the projection $\FF|_{S}$. The one-inclusion graph of $\FF|_{S}$ is a graph with the set of vertexes $\FF|_{S}$.  Vertexes $f$ and $g$ are connected with an edge iff 
\[
|\{x \in S: f(x) \neq g(x)\}| = 1,
\]
that is if $f$ and $g$ differ on exactly one point in $S$. One-inclusion graphs induced by VC classes satisfy the following remarkable property.
\begin{lem}{\cite{HLW, Haussler95}}
\label{orient}
For any finite $S \subset \X$ the edges of the one-inclusion graph induced by a class $\FF$ of VC dimension $d$ can be oriented in a way such that the out-degree of any vertex is at most $d$.
\end{lem}
Although we do not present the proof of this fact, we add a short remark about the simplest known way to show it. Using the \emph{shifting method} it is possible to prove that the edge densities  of the one-inclusion graph as well as any of its subgraphs are bounded by $d$. Finally, a standard application of Hall's marriage theorem implies Lemma \ref{orient}. We refer to \cite{Haussler95} for more details and discussions of this proof technique.

We are given an i.i.d. sample $T = ((x_1, f^*(x_1)), \ldots, (x_{n + 1}, f^*(x_{n + 1}))$. Denote by $T^{(i)}$ the sample $T$ with the hidden $i$-th label, that is 
\[
T^{(i)} = ((x_1, f^*(x_1)), \ldots, , (x_i, *), \ldots, (x_{n + 1}, f^*(x_{n + 1}))).
\]
Our aim is to predict the unknown label of $x_i$ given the sample $T{^{(i)}}$.
Fix $S = \{x_{1}, \ldots, x_{n +1}\}$ and consider the orientation of the one-inclusion graph induced by $\FF$ such that the out-degree of any vertex is at most $d$. We are ready to define the \emph{one-inclusion graph algorithm}. Given $T^{(i)}$ we define its output classifier $\hat{f}_{T^{(i)}}$ by $\hat{f}_{T^{(i)}}(x_j) = f^*(x_j)$ for $j \neq i$. Now, if there is only one $g \in \FF|_S$ such that it is consistent with labelled points of $T_i$, then we define $\hat{f}_{T^{(i)}}(x_i) = g(x_i)$. Observe that in this case $g(x_i) = f^*(x_i)$. Finally, if there are two distinct $g, h \in \FF|_S$ such that they are consistent with labelled points of $T_i$, we define $\hat{f}_{T^{(i)}}(x_i) = g(x_i)$ if in the oriented one-inclusion graph the edge is pointing from $h$ to $g$, and $\hat{f}_{T^{(i)}}(x_i) = h(x_i)$ if this edge is pointing from $g$ to $h$.

Finally, define the \emph{leave-one-out error} by
\[
\text{LOO} = \frac{1}{n + 1}\sum\limits_{i = 1}^{n + 1}\ind[\hat{f}_{T^{(i)}}(x_i) \neq f^*(x_i)].
\]
It follows from the definition of the one-inclusion graph algorithm that 
\[
\frac{1}{n + 1}\sum\limits_{i = 1}^{n + 1}\ind[\hat{f}_{T^{(i)}}(x_i) \neq f^*(x_i)] \le \frac{\text{out-degree of } f^*|_S}{n + 1} \le \frac{d}{n + 1}.
\]
At the same time, we have $\E\ \text{LOO} = \E_{T^{(1)}} \rho(\hat{f}_{T^{(1)}}, f^{*}) := \E \rho(\hat{f}, f^{*})$. This implies the desired risk bound
\[
\E \rho(\hat{f}, f^{*}) \le  \frac{d}{n + 1}.
\]

\begin{thebibliography}{DSS}
\bibitem{Alexander87} K. S. Alexander. {\it Rates of growth and sample moduli for weighted empirical processes indexed by sets.}  Probability Theory and Related Fields, 75:379--423, 1987.

\bibitem{Aro} B. Aronov, E. Ezra, and M. Sharir. {\it Small-size $\epsilon$-nets for axis-parallel rectangles and boxes},
SIAM J. Comput., 39(7):3248--3282, 2010.

\bibitem{Balcan13} M.F. Balcan , P. M. Long. {\it Active and passive learning of linear separators under log-concave distributions.} In Proceedings of the 26th Conference on Learning Theory, 2013.

\bibitem{Benedek91} G. Benedek, A Itai. {\it Learnability with respect to a fixed
distribution}. Theoretical Computer Science, 86, 377--389, 1991.

\bibitem{Bousquet05} O. Bousquet, S. Boucheron, G. Lugosi. {\it Introduction to Statistical Learning Theory.} Advanced Lectures on Machine Learning, LNCS, volume 3176, 2004.

\bibitem{Boucheron05}
S. Boucheron, O. Bousquet, G. Lugosi. \emph{Theory of classification: a survey of recent
advances.} ESAIM: Probability and Statistics, 9:323--375, 2005.

\bibitem{Bartlett05} P. L. Bartlett, O. Bousquet, S. Mendelson. {\it Local Rademacher
Complexities.} The Annals of Statistics, 33(4):1497--1537, 08, 2005.

\bibitem{Bshouty09}
N. H. Bshouty, Y. Li, P. M. Long, {\it Using the doubling dimension to analyze the generalization
of learning algorithms}, Journal of Computer and System Sciences 75, N6, 323--335, 2009.

\bibitem{Chan} T. M. Chan, E. Grant, J. Könemann, and M. Sharpe. {\it Weighted capacitated, priority, and
geometric set cover via improved quasi-uniform sampling}, In Proceedings of Symposium on Discrete Algorithms (SODA), 2012.

\bibitem{Chaz92} B. Chazelle. {\it A note on Haussler’s packing lemma.} Geometric Discrepancy: An
Illustrated Guide, 1992.

\bibitem{Chon92} D. Cohn, L. Atlas, R. Ladner. {\it Improving generalization with active learning}. Machine learning, 1994.

\bibitem{Doliwa14}
T. Doliwa, G. Fan, H. U. Simon, S. Zilles {\it Recursive Teaching Dimension, VC-Dimension and Sample
Compression}. Journal of Machine Learning Research 15, 3107-3131, 2014.

\bibitem{Dudley82}
R. M. Dudley. {\it A Course on Empirical Processes.} Ecole d'été de probabilités de St.-Flour, 1982.

\bibitem{Dutta15} K. Dutta, E. Ezra, A. Ghosh. {\it Two proofs for shallow packings.} Discrete and Computational Geometry,  Volume 56, Issue 4, pp 910--93, 2016.

\bibitem{Gine06}
	E. Gin\'e, V. Koltchinskii. {\it Concentration inequalities and asymptotic results for ratio
type empirical processes.} The Annals of Probability, 34(3):1143--1216, 2006.

\bibitem{Goldman95}
S. A. Goldman and M. J. Kearns. {\it On the complexity of teaching.} Journal of
Computer and System Sciences, 50(1):20-31, 1995.

\bibitem{Gupta03}
A. Gupta, R. Krauthgamer, and J. R. Lee. {\it Bounded geometries, fractals,
and low-distortion embeddings}, FOCS, 2003.

\bibitem{HW}D. Haussler, E. Welzl, {\it Epsilon-nets and simplex range queries}, Discrete \& Computational
Geometry, 127--151, 1987.

\bibitem{HLW}
D. Haussler, N. Littlestone, M. Warmuth. {\emph{Predicting
$\{0, 1\}$-functions on randomly drawn points. Information and Computation}, 115:248--292, 1994.}

\bibitem{Haussler95}
D. Haussler. {\it Sphere packing numbers for subsets of the boolean n-cube with bounded Vapnik-Chervonenkis dimension.} J. Comb. Theory Ser. A, 69(2):217--232, 1995.

\bibitem{KPW} J. Koml\'os, J. Pach, G.J. Woeginger, {\it Almost tight bounds for epsilon-nets}, Discrete \& Computational Geometry 7 (1992), 163--173.

\bibitem{KMP17} A. Kupavskii, N. Mustafa, J. Pach, \textit{Lower bounds for $\epsilon$-nets},  SoCG'2016.

\bibitem{Hanneke14}
S. Hanneke. {\it Theory of Disagreement-Based Active Learning,} Foundation and trends in machine learning, 2014.

\bibitem{Hanneke15}
S. Hanneke, L. Yang, {\it Minimax Analysis of Active Learning}, Journal of Machine Learning Research 12, 3487--3602, 2015.

\bibitem{Han15}
S. Hanneke, {\it The Optimal Sample Complexity of PAC Learning}, Journal of Machine Learning Research 17, 1--15, 2016.

\bibitem{Matousek}
J. Matou\u{s}ek, {\it Geometric discrepancy}, Springer Berlin Heidelberg, 1999

\bibitem{Lecam73}
L. M. Le Cam. {\it Convergence of estimates under dimensionality restrictions}. Annals of Statistics
1, 38--53, 1973.

\bibitem{Mendelson17}
S. Mendelson. {\it `Local' vs. `global' parameters -- breaking the Gaussian complexity barrier}. Annals of Statistics, 2017.

\bibitem{Mustafa16} N. Mustafa. {\it A Simple Proof of the Shallow Packing Lemma}. Discrete and Computational Geometry, Volume 55, Issue 3, pp 739--743, 2016.

\bibitem{MDG} N. Mustafa, K. Dutta, A. Ghosh,
{\it A simple proof of optimal epsilon-nets}. Combinatorica, Volume 38, Issue 5, pp 1269--1277,  2018.

\bibitem{MV} N. Mustafa, K. Varadarajan, {\it Epsilon-approximations and epsilon-nets}, arXiv:1702.03676, 2017

\bibitem{VC} V. Vapnik,  A. Chervonenkis. {\it Theory of Pattern Recognition}. Nauka, Moscow, 1974.

\bibitem{Var} K. Varadarajan, {\it Weighted geometric set cover via quasi uniform sampling}, In Proceedings of the Symposium on Theory of Computing (STOC), pages 641--648, New York, USA, 2010. 

\bibitem{Zhang14}
C. Zhang, K. Chaudhuri. {\it Beyond Disagreement-based Agnostic Active Learning}. NIPS, 2014.

\bibitem{Zhivotovskiy16}
N. Zhivotovskiy, S. Hanneke. {\it Localization of VC classes: beyond local Rademacher Complexities}. Theoretical Computer Science, 2018.

\bibitem{Zhivotovskiy17}
N. Zhivotovskiy. {\it Optimal learning via local entropies and sample compression}. Conference on Learning Theory, 2017.

\end{thebibliography}
\end{document}